\documentclass[conference]{IEEEtran}
\IEEEoverridecommandlockouts
\usepackage{cite}
\usepackage{amsmath,amssymb,amsfonts}
\usepackage{algorithmic}
\usepackage{graphicx}
\usepackage{textcomp}
\usepackage{xcolor}
\def\BibTeX{{\rm B\kern-.05em{\sc i\kern-.025em b}\kern-.08em
    T\kern-.1667em\lower.7ex\hbox{E}\kern-.125emX}}

\usepackage{subfigure}
\usepackage{tikz}
\usepackage{amsthm}
\usepackage{cleveref}
\usepackage{xspace}
\usepackage{cite}

\usepackage{float}

\definecolor{darkgreen}{rgb}{0.0,0.5,0.0}

\newtheoremstyle{newstyle1}{}{}{\itshape}{}{}{}{0em}{\textsc{\thmname{#1}}\thmnumber{ #2}. }
\theoremstyle{newstyle1}

\newtheorem{definition}{Definition}
\newtheorem{algorithm}{Algorithm}
\newtheorem{observation}{Observation}
\newtheorem{lemma}{Lemma}
\newtheorem{corollary}{Corollary}
\newtheorem{theorem}{Theorem}

\newcommand{\imped}{impedensable\xspace}
\newcommand{\Imped}{Impedensable\xspace}
\newcommand{\gmrit}{GMRIT\xspace}

\usepackage[switch]{lineno}

\begin{document}

\title{Tolerance to Asynchrony of an Algorithm for Gathering Myopic Robots on an Infinite Triangular Grid\thanks{\textbf{Technical report of the paper to appear in the 19th European Dependable Computing Conference (EDCC 2024)}.}}

\author{\IEEEauthorblockN{Arya Tanmay Gupta and Sandeep S Kulkarni}
\IEEEauthorblockA{\textit{Computer Science and Engineering, Michigan State University} \\
\texttt{\{atgupta,sandeep\}@msu.edu}}
}

\maketitle

\begin{abstract}
In this paper, we study the problem of gathering distance-1 myopic robots on an infinite triangular grid. We show that the algorithm developed by Goswami et al. (SSS, 2022) is lattice-linear (cf. Gupta and Kulkarni, SRDS 2023). This implies that a distributed scheduler, assumed therein, is not required for this algorithm: it runs correctly in asynchrony. It also implies that the algorithm works correctly even if the robots are equipped with a unidirectional \textit{camera} to see the neighbouring robots (rather than an omnidirectional one, which would be required under a distributed scheduler). Due to lattice-linearity, we can predetermine the point of gathering. We also show that this algorithm converges in $2n$ rounds, which is lower than the complexity ($2.5(n+1)$ rounds) that was shown in Goswami et al.
\end{abstract}


\section{Introduction}

A parallel processing system contains multiple processes running to solve one problem. 
Each process may store some data. While the system executes, all processes perform executions based on the information that they read from each other. The goal of these processes is to reach a state where the problem is deemed solved. 
To ensure correct execution, synchronization primitives are deployed.

Synchronization ensures that the processes read the latest copy of data.
In some of the synchronization models, a selected process acts as a scheduler for the rest of the processes. A \textit{scheduler/daemon} is a process whose function is to choose one, some, or all processes in a time step, throughout the execution, so that the selected processes can evaluate their guards and take the corresponding action. A \textit{central scheduler} chooses only one process per time step. A \textit{distributed scheduler} one or more processes, possibly arbitrarily, per time step. A \textit{synchronous scheduler} chooses all the processes in each time step.

Enforcing synchronization introduces an overhead, which can be very costly in terms of computational resources and time. For this reason, the community is interested in developing algorithms that require minimal synchronization. The recent introduction of lattice-linearity has shown that algorithms can be developed in such a way that the processes can execute asynchronously. 
Concepts of lattice-linearity can be used to develop new lattice-linear algorithms or to study if existing algorithms exploit lattice-linearity.
If existing algorithms can exploit lattice-linearity, it would imply that we can eliminate the synchronization assumptions -- thereby saving the time and resources involved in synchronization during their execution. 
For example, lattice-linearity was found to be exploited by Gale-Shapley's algorithm for stable marriage \cite{Garg2020}, Johnson's algorithm for shortest paths in graphs \cite{Garg2020} and by Gale’s top trading cycle algorithm for housing market \cite{Garg2021}. We introduced lattice-linear algorithms in \cite{Gupta2023} for problems that contain multiple optimal states, whereas the model studied in \cite{Garg2020} only allows one optimal state.

In this paper, we study the problem of gathering distance-1 myopic robots on an infinite triangular grid. Authors of \cite{Goswami2022} showed this algorithm to work correctly under a distributed scheduler. We demonstrate that this algorithm is lattice-linear.
Hence, it will run correctly even if the robots run in asynchrony. 

\subsection{Contributions of the paper}

\begin{itemize}
    \item We show that the algorithm developed by Goswami et al. \cite{Goswami2022} is lattice-linear. Hence, this algorithm will run correctly even if the robots run in asynchrony where the robots can execute on old information about other robots.
    Because of lattice-linearity, this algorithm works correctly even if the robots are equipped with a unidirectional \textit{camera} to see neighbouring robots.
    Authors of \cite{Goswami2022} assumed a distributed scheduler, which would require an omnidirectional camera, capable to get fresh values from all neighbouring locations.
    \item Lattice-linearity follows that the moves of the robots are predictable. 
    This allows us to show tighter bounds to the arena traversed by the robots under the algorithm.
    As a consequence of tighter bounds on this arena, (1) we obtain a better convergence time bound for this algorithm, which is lower than that showed in \cite{Goswami2022}, and (2) we show that the gathering point of the robots can be uniquely determined from the initial, or any intermediate, state. We show that this algorithm converges in $2n$ rounds, which is lower as compared to the time complexity bound ($2.5(n+1)$ rounds) shown in \cite{Goswami2022}.
\end{itemize}

\subsection{Organization of the paper}

In \Cref{section:preliminaries}, we describe the notations and definitions that we utilize in this paper. In \Cref{section:gmrit}, we describe the problem of gathering distance-1 myopic robots on an infinite triangular grid. In \Cref{section:algorithm-gsgs}, we describe the lattice-linearity of the algorithm in \cite{Goswami2022}; we also show some tighter bounds on the properties of this algorithm.
We provide a revised algorithm for this problem in \Cref{section:gsgs-new-algo}.
We discuss the related work in \Cref{section:literature} and conclude in \Cref{section:conclusion}.

\section{Preliminaries}\label{section:preliminaries}

A parallel/distributed algorithm consists of nodes where each node is associated with a set of variables. A \textit{global state}, say $s$, is obtained by assigning each variable of each node a value from its respective domain. $s$ is represented as a vector, where $s[i]$ itself is a vector of the variables of node $i$.
$S$ denotes the \textit{state space}, which is the set of all global states that a given system can obtain. 

Each node is associated with actions. Each action at node $i$ checks the values of other nodes and updates its own variables. A \textit{rule} is of the form $g\longrightarrow a_c$, where $g$ is the \textit{guard} (a Boolean predicate involving variables of $i$ and other nodes) and $a_c$ is an \textit{instruction/action} that updates the variables of $i$
atomically.

A \textit{move} is an event where a node 
changes its state when one of its guards evaluates to true. A \textit{round} is a sequence of events in which every node evaluates its guards at least once, and makes a move accordingly.

An algorithm $A$ is \textit{self-stabilizing} with respect to the subset $S_o$ of $S$ iff (1) \textit{convergence}: starting from an  arbitrary state, any sequence of computations of $A$ reaches a state in $S_o$, and (2) \textit{closure}: any computation of $A$ starting from $S_o$ always stays in $S_o$. 
We assume $S_o$ to be the set of \textit{optimal} states: the system is deemed converged once it reaches a state in $S_o$. $A$ is a \textit{silent} self-stabilizing algorithm if no node makes a move once a state in $S_o$ is reached.

\subsection{Execution without Synchronization}

Typically, we view the \textit{computation} of an algorithm as a sequence of global states $\langle s_0, s_1, \cdots\rangle$, where $s_{t+1}, t\geq 0,$ is obtained by executing some action by one or more nodes in $s_t$.  
For the sake of discussion, assume that only node $i$ executes in state $s_t$. 
The computation prefix uptil $s_{t}$ is $\langle s_0, s_1, \cdots, s_t\rangle$. The state that the system traverses to after $s_t$ is $s_{t+1}$.
Under proper synchronization, $i$ would evaluate its guards on the \textit{current} local states of its neighbours in $s_t$.

To understand the execution in asynchrony, let $x(s)$ be the value of some variable $x$ 
in state $s$. 
If $i$ executes in asynchrony, then it views the global state that it is in to be $s'$, 
where $x(s')\in\{ x(s_0), x(s_1), \cdots, x(s_t) \}.$
In this case, $s_{t+1}$ is evaluated as follows.
If all guards in $i$ evaluate to false, then the system will continue to remain in state $s_t$, i.e., $s_{t+1} = s_{t}$.
If a guard $g$ evaluates to true then $i$ will execute its corresponding action $a_c$.
Here, we have the following observations:
(1) $s_{t+1}[i]$ is the state that $i$ obtains after executing an action in $s'$, and (2) $\forall j\neq i$: $s_{t+1}[j] = s_t[j]$.

The model described above allows nodes to read old values of other nodes arbitrarily. In this paper, however, we require that the nodes (robots) take a \textit{snapshot} of their surroundings after every move. However, doing so does not require synchronization. 

\subsection{Embedding a $\prec$-lattice among Global States}

Let $s$ denote a global state, and let $s[i]$ denote the state of node $i$ in $s$. First, we define a total order $\prec_l$; all local states of a node $i$ are totally ordered under $\prec_l$. 
Using $\prec_l$, we define a partial order $\prec_g$ among global states as follows. 

We say that $s \prec_g s^\prime$ iff $(\forall i: s[i]=s'[i]\lor s[i]\prec_l s'[i]) \land (\exists i:s[i]\prec_ls'[i])$.
Also, $s=s'$ iff $\forall i: s[i] = s'[i]$. 
For brevity, we use $\prec$ to denote $\prec_l$ and $\prec_g$: $\prec$ corresponds to $\prec_l$ while comparing local states, and $\prec$ corresponds to $\prec_g$ while comparing global states. 
We also use the symbol `$\succ$' which is such that $s\succ s'$ iff $s' \prec s$.
Similarly, we use symbols `$\preceq$' and `$\succeq$'; e.g., $s\preceq s'$ iff  $s=s' \vee s \prec s'$.
We call the lattice, formed from such partial order, a \textit{$\prec$-lattice}.

\begin{definition}\label{definition:<-lattice}
    $\prec$-\textit{lattice}. 
    Given a total relation $\prec_l$ that orders the states visited by a node $i$ (for each $i$) the $\prec$-lattice corresponding to $\prec_l$ is defined by the following partial order:
    $s \prec_g s'$ iff $(\forall i \ \ s[i] \preceq_l s'[i]) \wedge (\exists i \ \ s[i] \prec_l s'[i])$.
\end{definition}

In the $\prec$-lattice discussed above, we can define the meet and join of two states in the standard way: the meet (respectively, join), of two states $s_1$ and $s_2$ is a state $s_3$ where $\forall i, s_3[i]$ is equal to $min(s_1[i], s_2[i])$ (respectively, $max(s_1[i], s_2[i])$), where $\min(x, y) = \min(y, x)=x$ iff $(x\prec_l y \lor x=y)$, and 
$\max(x, y) = \max(y, x)=y$ iff $(y\succ_l x \lor y=x)$.
For $s_1$ and $s_2$, their meet (respectively, join) has paths to (respectively, is reachable from) both $s_1$ and $s_2$.

A $\prec$-lattice, embedded in the state space, is useful for permitting the algorithm to execute asynchronously.
Under proper constraints on the structure of $\prec$-lattice, convergence can be ensured. 

\subsection{Lattice-Linear Problems}

A \textit{lattice-linear problem} $P$ can be represented by a predicate $\mathcal{P}$ such that if any node $i$ is violating $\mathcal{P}$ in a state $s$, then it must change its state. Otherwise, the system will not satisfy $\mathcal{P}$.
Let $\mathcal{P}(s)$ be true iff state $s$ satisfies $\mathcal{P}$. A node violating $\mathcal{P}$ in $s$ is called an \textit{\imped} node (an \textit{impediment} to progress if does not execute, \textit{indispensable} to execute for progress). Formally,

\begin{definition}\label{definition:impedensable-node}\cite{Garg2020} \textit{\Imped node.} $\textsc{\Imped}(i,s,\mathcal{P})\equiv \lnot \mathcal{P}(s)$ $\land$ $(\forall s'\succ s:s'[i]=s[i]\implies\lnot \mathcal{P}(s'))$. \end{definition}

If a node $i$ is \imped in state $s$, then in any state $s'$ such that $s'\succ s$, if the state of $i$ remains the same, then the algorithm will not converge.
Thus, predicate $\mathcal{P}$ induces a total order among the local states visited by a node, for all nodes. Consequently, the discrete structure that gets induced among the global states is a $\prec$-lattice, as described in \Cref{definition:<-lattice}. 
We say that $\mathcal{P}$, satisfying \Cref{definition:impedensable-node}, is \textit{lattice-linear} with respect to that $\prec$-lattice.

There can be multiple arbitrary lattices that can be induced among the global states. A system cannot guarantee convergence while traversing an arbitrary lattice. To guarantee convergence, we design the predicate $\mathcal{P}$ such that it fulfils some properties, and guarantees reachability to an optimal state. $\mathcal{P}$ is used by the nodes to determine if they are \imped, using \Cref{definition:impedensable-node}.
Thus, in any suboptimal global state, there will be at least one \imped node. Formally,

\begin{definition}\label{definition:ll-predicate}\cite{Garg2020}\textit{Lattice-linear predicate.}
    $\mathcal{P}$ is a lattice-linear predicate with respect to a $\prec$-lattice induced among the global states iff $\forall s\in S: \lnot\mathcal{P}(s) \Rightarrow \exists i:\textsc{\Imped}(i,s,\mathcal{P})$.
\end{definition}

Now we complete the definition of lattice-linear problems. In a lattice-linear problem $P$, given any suboptimal global state $s$, $P$ specifies all and the only nodes which cannot retain their local states. 
$\mathcal{P}$ is thus designed conserving this nature of the subject problem $P$, following Definitions \ref{definition:impedensable-node} and \ref{definition:ll-predicate}.

\begin{definition}\label{definition:ll-problem}
    \textit{Lattice-linear problems}.
    A problem $P$ is lattice-linear 
    iff there exists a predicate $\mathcal{P}$ and a $\prec$-lattice such that
    
    \begin{itemize}
        \item $P$ is deemed solved iff the system reaches a state where $\mathcal{P}$ is true,
        \item $\mathcal{P}$ is lattice-linear with respect to the $\prec$-lattice induced among the states in $S$, i.e., $\forall s: \neg \mathcal{P}(s) \Rightarrow \exists i:\textsc{\Imped}(i,s,\mathcal{P})$, and
        \item $\forall s:(\forall i:\textsc{\Imped}(i,s,\mathcal{P})\Rightarrow (\forall s':\mathcal{P}(s')\Rightarrow s'[i]\neq s[i]))$.
    \end{itemize}
\end{definition}

\noindent\textit{Remark:} A $\prec$-lattice, induced under $\mathcal{P}$, allows asynchrony because if a node, reading old values, reads the current state $s$ as $s'$, then $s'\prec s$. So $\lnot\mathcal{P}(s')\Rightarrow \lnot\mathcal{P}(s)$ because $\textsc{\Imped}(i,s',\mathcal{P})$ and $s'[i]=s[i]$.

\begin{definition}\label{definition:ssll-problem}
    \textit{Self-stabilizing lattice-linear predicate}.
    Continuing from \Cref{definition:ll-problem},
    $\mathcal{P}$ is a self-stabilizing lattice-linear predicate if and only if the supremum of the lattice, that $\mathcal{P}$ induces, is an optimal state, i.e., $\mathcal{P}(supremum(S))=true$.
\end{definition}

\noindent Note that a self-stabilizing lattice-linear predicate $\mathcal{P}$ can also be true in states other than the supremum of the $\prec$-lattice. 

\subsection{Lattice-Linear Algorithms}\label{subsection:lla}

Certain problems are \textit{non-lattice-linear problems}. In those problems, given a suboptimal global state, the problem does not specify a specific set of nodes to change their state. In such problems, there are instances in which the \imped nodes cannot be determined naturally, i.e., in those instances
$\exists s :\lnot\mathcal{P}(s) \wedge   (\forall i : \exists s' : \mathcal{P}(s')\land s[i]=s'[i]$).
For such problems, $\prec$-lattices may be induced algorithmically, through \textit{lattice-linear algorithms}.

\begin{definition}\label{definition:ll-algos}\cite{Gupta2022,Gupta2023}\textit{Lattice-linear algorithms (LLA)}.
    Algorithm $A$ is an LLA for a problem $P$, iff there exists a predicate $\mathcal{P}$ and $A$ induces a $\prec$-lattice among the states of $S_1, ..., S_w \subseteq S (w\geq 1)$, such that
    \begin{itemize}
        \item State space $S$ of $P$ contains mutually disjoint lattices, i.e.
        \begin{itemize}
            \item $S_1, S_2, \cdots, S_w\subseteq S$ are pairwise disjoint.
            \item $S_1 \cup S_2 \cup \cdots \cup S_w$ contains all the reachable states (starting from a set of initial states, if specified; if an arbitrary state can be an initial state, then $S_1 \cup S_2 \cup \cdots \cup S_w=S$).
        \end{itemize}
        \item Lattice-linearity is satisfied in each subset under $\mathcal{P}$, i.e., 
        \begin{itemize}
            \item $P$ is deemed solved iff the system reaches a state where $\mathcal{P}$ is true
            \item $\forall k$, $1 \leq k \leq w$, 
            $\mathcal{P}$ is lattice-linear with respect to the partial order induced in $S_k$ by $A$, i.e., $\forall s\in S_k: \lnot\mathcal{P}(s) \Rightarrow \exists i \ \
            \textsc{\Imped}(i,s,\mathcal{P})$.
        \end{itemize}
    \end{itemize}
\end{definition}

In the next section, we discuss how the gathering problem of distance-1 myopic robots is not a lattice-linear problem. The lattice-linear algorithm developed for this problem imposes the lattice structure among the global states, that are reachable by the system of robots.

\begin{definition}\textit{Self-stabilizing LLA}.
    From \Cref{definition:ll-algos}, $A$ is self-stabilizing only if $S_1 \cup S_2 \cup \cdots \cup S_w=S$ and $\forall k:1\leq k\leq w$, $\mathcal{P}(supremum(S_k))=true$.
\end{definition}

\section{Gathering of Distance-1 Myopic Robots on Infinite Triangular Grid (\gmrit)}\label{section:gmrit}

This paper focuses on the problem where the input is a swarm of robots
with minimal capabilities. 
Each robot is present at a vertex on an infinite triangular grid. 
In the initial global state, the robots form a connected graph on the underlying grid. The robots agree on an axis (i.e. a direction and its orientation).
The robots can move only across one edge at a time.
Each robot is myopic, i.e., it can only sense if another robot is present at an adjacent vertex.
Robots do not have an abiliy to \textit{communicate} with each other. 
Under these constraints, it is required that all robots gather at one point.

\newcounter{diags}
\newcounter{rdiag}
\newcounter{vert}
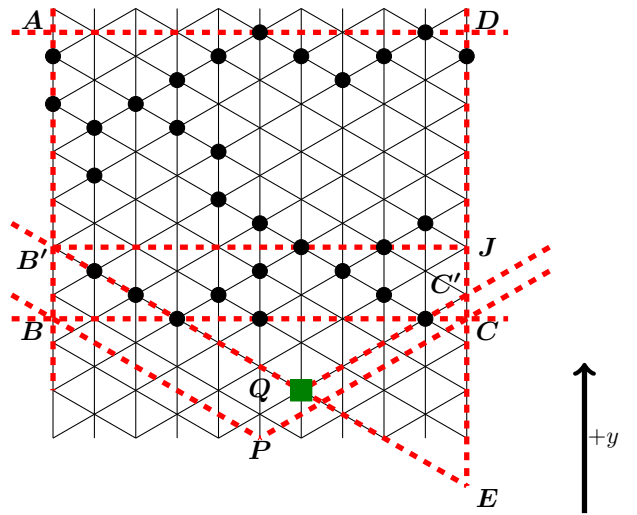
\begin{figure}[ht]
    \centering
    \subfigure{
    \begin{tikzpicture}[x=1.1cm,y=0.6351cm] 
        \foreach \diag in {0,...,4}{
            \draw (\diag,0) -- (5,-5+\diag);
        }
        \foreach \diag in {8,...,5}{
            \draw (0,-9+\diag) -- (5,\diag-14);
        }

        \foreach \diag in {4,...,1}{
            \draw (0,-9+\diag) -- (\diag,-9);
        }
        
        \foreach \rdiag in {0,...,4}{
            \draw (\rdiag,-9) -- (5,-\rdiag-4);
        }
        \foreach \rdiag in {8,...,5}{
            \draw (0,-\rdiag) -- (5,5-\rdiag);
        }
        \foreach \rdiag in {4,...,1}{
            \draw (0,-\rdiag) -- (\rdiag,0);
        }
        
        \foreach \vert in {0,...,10}{
            \draw (\vert*.5,0) -- (\vert*.5,-9);
        }

        \draw[dashed, red, line width=2pt] (0,0) -- (0,-8);
        \draw[dashed, red, line width=2pt] (5,0) -- (5,-10);

        \draw[dashed, red, line width=2pt] (-.5,-.5) -- (5.5,-.5);
        \draw[dashed, red, line width=2pt] (-.5,-6.5) -- (5.5,-6.5);
        
        \draw[dashed, red, line width=2pt] (-.5,-4.5) -- (3,-8) -- (5,-10);
        \draw[dashed, red, line width=2pt] (6,-5) -- (3,-8);
        \draw[dashed, red, line width=2pt] (-.5,-6) -- (2.5,-9);
        \draw[dashed, red, line width=2pt] (6,-5.5) -- (2.5,-9);
        
        \draw[dashed, red, line width=2pt] (0,-5) -- (5,-5);
        
        \node[font=\boldmath] at (-.25,-.25) {$A$};
        \node[font=\boldmath] at (-.25,-6.75) {$B$};
        \node[font=\boldmath] at (5.25,-6.75) {$C$};
        \node[font=\boldmath] at (5.25,-.25) {$D$};
        \node[font=\boldmath] at (2.5,-9.25) {$P$};
        
        \node[font=\boldmath] at (-.25,-5.25) {$B'$};
        \node[font=\boldmath] at (4.75,-5.75) {$C'$};
        \node[font=\boldmath] at (2.5,-8) {$Q$};

        \node[font=\boldmath] at (5.25,-10.25) {$E$};
        \node[font=\boldmath] at (5.25,-5) {$J$};
        
        \node [circle, inner sep=2pt, fill=black, draw=black] at (0,-1) {};
        \node [circle, inner sep=2pt, fill=black, draw=black] at (0,-2) {};
        \node [circle, inner sep=2pt, fill=black, draw=black] at (.5,-2.5) {};
        \node [circle, inner sep=2pt, fill=black, draw=black] at (.5,-3.5) {};
        \node [circle, inner sep=2pt, fill=black, draw=black] at (.5,-5.5) {};
        \node [circle, inner sep=2pt, fill=black, draw=black] at (1,-2) {};
        \node [circle, inner sep=2pt, fill=black, draw=black] at (1,-6) {};
        \node [circle, inner sep=2pt, fill=black, draw=black] at (1.5,-1.5) {};
        \node [circle, inner sep=2pt, fill=black, draw=black] at (1.5,-2.5) {};
        \node [circle, inner sep=2pt, fill=black, draw=black] at (1.5,-6.5) {};
        \node [circle, inner sep=2pt, fill=black, draw=black] at (2,-1) {};
        \node [circle, inner sep=2pt, fill=black, draw=black] at (2,-3) {};
        \node [circle, inner sep=2pt, fill=black, draw=black] at (2,-4) {};
        \node [circle, inner sep=2pt, fill=black, draw=black] at (2,-6) {};
        \node [circle, inner sep=2pt, fill=black, draw=black] at (2.5,-.5) {};
        \node [circle, inner sep=2pt, fill=black, draw=black] at (2.5,-4.5) {};
        \node [circle, inner sep=2pt, fill=black, draw=black] at (2.5,-5.5) {};
        \node [circle, inner sep=2pt, fill=black, draw=black] at (2.5,-6.5) {};
        \node [circle, inner sep=2pt, fill=black, draw=black] at (3,-1) {};
        \node [circle, inner sep=2pt, fill=black, draw=black] at (3,-5) {};
        \node [circle, inner sep=2pt, fill=black, draw=black] at (3.5,-1.5) {};
        \node [circle, inner sep=2pt, fill=black, draw=black] at (3.5,-5.5) {};
        \node [circle, inner sep=2pt, fill=black, draw=black] at (4,-1) {};
        \node [circle, inner sep=2pt, fill=black, draw=black] at (4,-5) {};
        \node [circle, inner sep=2pt, fill=black, draw=black] at (4,-6) {};
        \node [circle, inner sep=2pt, fill=black, draw=black] at (4.5,-.5) {};
        \node [circle, inner sep=2pt, fill=black, draw=black] at (4.5,-4.5) {};
        \node [circle, inner sep=2pt, fill=black, draw=black] at (4.5,-6.5) {};
        \node [circle, inner sep=2pt, fill=black, draw=black] at (5,-1) {};

        \node [rectangle, inner sep=4pt, fill=darkgreen, draw=darkgreen] at (3,-8) {};

    \end{tikzpicture}
    }
    \subfigure{
        \begin{tikzpicture}[scale=.5,every node/.style={scale=.9}]
            \draw[->, line width=2pt] (0,0) -- (0,4);
            \node at (.5,2) {$+y$};
        \end{tikzpicture}
    }
    \caption{Robots on an infinite triangular grid: one on every round highlighted vertex.}
    \label{figure:infinite-triangular-grid}
\end{figure}

\subsection{Problem Statement}
The input is a global state $s$ that describes the location of $n$ robots placed on the grid points of an infinite triangular grid $G$ such that the robots form a connected graph. 
The \gmrit problem requires that all robots gather at one vertex of $G$ and stay forever at that vertex subject to the following constraints:

\begin{itemize}
    \item \textit{Visibility:}
    A robot can only determine if another robot is present in a neighbouring location. It cannot exchange data with another robot.
    \item \textit{One Axis Agreement:}
    All robots agree on one axis and the orientation of that axis. (cf. $y$-axis as shown in \Cref{figure:infinite-triangular-grid}).
\end{itemize}

All robots are independent and identical from a physical and computational perspective and do not have an ID.
They are oblivious to the coordinates of their location on the infinite triangular grid $G$.
Observe that we can allow a global state $s$ to be a multiset of vertices, each of which is the location of a robot.
For $s$, its \textit{visibility graph} is the subgraph of $G$ induced by the set of vertices in $s$.

Notice that by definition of the problem statement, a solution to \gmrit will provide silent self-stabilization.
An instance of \gmrit
is shown in \Cref{figure:infinite-triangular-grid}. 
Here, each round highlighted vertex represents a robot. 
Observe that the visibility graph of this global state is a connected graph.

Next, we discuss how \gmrit problem is not a lattice-linear problem. This can be illustrated by a system containing two robots $x_1$ and $x_2$ present at locations $l_1$ and $l_2$ ($l_1$ and $l_2$ are different vertices on the same edge) on $G$. In such a system $x_1$ can move to $l_2$, in which case, $x_2$ is not \imped, or otherwise, $x_2$ can move to $l_1$, in which case, $x_1$ is not \imped. Hence, no specific robot can be deemed \imped, though, the global state is suboptimal.

\subsection{Problem Specific Definitions}\label{subsection:gsgs-definitions}

Some of the definitions that we discuss in this subsection are from \cite{Goswami2022}. 
A \textit{horizontal layer} is a line perpendicular to the $y$-axis that passes through at least one robot.
The \textit{top layer} of a global state $s$ is a horizontal layer such that there is no horizontal layer above it (e.g., $AD$ in \Cref{figure:infinite-triangular-grid}).
\textit{Bottom layer} of $s$ is a horizontal layer such that there is no horizontal layer below it (e.g., $BC$ in \Cref{figure:infinite-triangular-grid}).

A \textit{vertical layer} is parallel to the $y$-axis such that it passes through at least one robot.
The \textit{left layer} of $s$ is the vertical layer such that there is no vertical layer on its left (e.g., $AB$ in \Cref{figure:infinite-triangular-grid}).
The \textit{right layer} of $s$ is the vertical layer such that there is no vertical layer on its right (e.g., $CD$ in \Cref{figure:infinite-triangular-grid}).

As seen in \Cref{figure:infinite-triangular-grid}, vertices in $G$ are intersections of three groups of parallel lines; one of these groups are lines parallel to the $y$-axis. We use $p$-axis (positive slope) and $n$ axis (negative slope) to denote the other group of parallel lines.
The \textit{positive slant} is a line
parallel to $p$-axis (e.g., $BP$ in \Cref{figure:infinite-triangular-grid}) and \textit{negative slant} is a line
parallel to $n$-axis (e.g., $CP$ in \Cref{figure:infinite-triangular-grid}).
The \textit{bottom $l2r$ slant} of $s$ is a negative slant that passes through a robot such that there is no negative slant on its left passing through a robot (e.g., $B'Q$ in \Cref{figure:infinite-triangular-grid}).
The \textit{bottom $r2l$ slant} of $s$ is a positive slant that passes through a robot such that there is no positive slant on its right passing through a robot (e.g., $C'Q$ in \Cref{figure:infinite-triangular-grid}).
Note that a negative slant and a positive slant can be imaginary, or a line in $G$.

The \textit{depth} of $s$ is the distance between its top layer and its bottom layer.
The \textit{width} of $s$ is defined as the distance between its left layer and right layer.

As shown in \Cref{figure:infinite-triangular-grid}, a polygon $ABPCD$ is a \textit{bounding polygon} of a global state $s$ if
(1) $AB$ and $CD$ are line segments of the left layer and the right layer of $s$ respectively, (2) $AD$ and $BC$ are line segments of the top layer and the bottom layer of $s$ respectively, and (3) $P$ is the point of intersection between the negative slant passing through $B$ and the positive slant passing through $C$.

Note that these definitions (top/bottom layer, etc.) are only used for discussion of the protocol and proofs. The robots are not aware of them. Similarly, the robots can distinguish between $up$ and $down$, but not $left$ and $right$.

\subsection{General Idea of the Algorithm}

A robot has six possible neighbouring locations. The naming convention for these locations is as shown in \Cref{figure:local-states} (a) \cite{Goswami2022}. Since each robot has 6 neighbouring locations, it can be in one of $2^6$ possible local states. Of these, the robot can move in only 11 states. Of these, 7 states are shown in \Cref{figure:local-states} (b) \cite{Goswami2022}. The other 4 states are mirror images of those shown in cases 2, 5, 6 and 7 in \Cref{figure:local-states} (b).

\begin{figure}[ht]
    \centering
    \subfigure[]{
    \begin{minipage}{.25\textwidth}
        \begin{tikzpicture}[x=1.5cm,y=.8660cm] 
        \node [circle, fill = black, inner sep = 2pt, label=left:$i$] (i) at (0,0) {};
        \node [circle, fill = black, inner sep = 1pt, label=below:$v_1$] (v1) at (0,-1) {};
        \node [circle, fill = black, inner sep = 1pt, label=left:$v_2^\ell$] (v2l) at (-.5,-.5) {};
        \node [circle, fill = black, inner sep = 1pt, label=right:$v_2^r$] (v2r) at (.5,-.5) {};
        \node [circle, fill = black, inner sep = 1pt, label=left:$v_3^\ell$] (v3l) at (-.5,.5) {};
        \node [circle, fill = black, inner sep = 1pt, label=right:$v_3^r$] (v3r) at (.5,.5) {};
        \node [circle, fill = black, inner sep = 1pt, label=above:$v_4$] (v4) at (0,1) {};

        \node [circle, draw = black, inner sep = 2pt] (v1) at (0,-1) {};
        \node [circle, draw = black, inner sep = 2pt] (v2l) at (-.5,-.5) {};
        \node [circle, draw = black, inner sep = 2pt] (v3l) at (-.5,.5) {};

        \draw (i) -- (v1); \draw (i) -- (v2l); \draw (i) -- (v2r); \draw (i) -- (v3l); \draw (i) -- (v3r);
        \draw (i) -- (v4);
        
        \draw (v1) -- (v2l) -- (v3l) -- (v4) -- (v3r) -- (v2r) -- (v1);
    \end{tikzpicture}
    \end{minipage}
    }
    \renewcommand{\thesubfigure}{(b)}
    \subfigure[]{
    \begin{minipage}{.35\textwidth}
        \subfigure{
        \begin{tikzpicture}[x=1cm,y=.5773cm] 
            \node [circle, fill = black, inner sep = 2pt] (i) at (0,0) {};
            \node [circle, draw = black, inner sep = 2pt] (v1) at (0,-1) {};
            \node [circle, fill = black, inner sep = .5pt] (v2l) at (-.5,-.5) {};
            \node [circle, fill = black, inner sep = .5pt] (v2r) at (.5,-.5) {};
            \node [circle, fill = black, inner sep = .5pt] (v3l) at (-.5,.5) {};
            \node [circle, fill = black, inner sep = .5pt] (v3r) at (.5,.5) {};
            \node [circle, fill = black, inner sep = .5pt] (v4) at (0,1) {};
    
            \draw (i) -- (v1); \draw (i) -- (v2l); \draw (i) -- (v2r); \draw (i) -- (v3l); \draw (i) -- (v3r);
            \draw (i) -- (v4);
            
            \draw (v1) -- (v2l) -- (v3l) -- (v4) -- (v3r) -- (v2r) -- (v1);
            
            \node at (0,-1.5) {$\ast$};
        \end{tikzpicture}
    }
    \subfigure{
        \begin{tikzpicture}[x=1cm,y=.5773cm] 
            \node [circle, fill = black, inner sep = 2pt] (i) at (0,0) {};
            \node [circle, draw = black, inner sep = 2pt] (v1) at (0,-1) {};
            \node [circle, draw = black, inner sep = 2pt] (v2l) at (-.5,-.5) {};
            \node [circle, fill = black, inner sep = .5pt] (v2r) at (.5,-.5) {};
            \node [circle, fill = black, inner sep = .5pt] (v3l) at (-.5,.5) {};
            \node [circle, fill = black, inner sep = .5pt] (v3r) at (.5,.5) {};
            \node [circle, fill = black, inner sep = .5pt] (v4) at (0,1) {};
    
            \draw (i) -- (v1); \draw (i) -- (v2l); \draw (i) -- (v2r); \draw (i) -- (v3l); \draw (i) -- (v3r);
            \draw (i) -- (v4);
            
            \draw (v1) -- (v2l) -- (v3l) -- (v4) -- (v3r) -- (v2r) -- (v1);
            
            \node at (0,-1.5) {$\ast$};
        \end{tikzpicture}
    }
    \subfigure{
        \begin{tikzpicture}[x=1cm,y=.5773cm] 
            \node [circle, fill = black, inner sep = 2pt] (i) at (0,0) {};
            \node [circle, fill = black, inner sep = .5pt] (v1) at (0,-1) {};
            \node [circle, draw = black, inner sep = 2pt] (v2l) at (-.5,-.5) {};
            \node [circle, draw = black, inner sep = 2pt] (v2r) at (.5,-.5) {};
            \node [circle, fill = black, inner sep = .5pt] (v3l) at (-.5,.5) {};
            \node [circle, fill = black, inner sep = .5pt] (v3r) at (.5,.5) {};
            \node [circle, fill = black, inner sep = .5pt] (v4) at (0,1) {};
    
            \draw (i) -- (v1); \draw (i) -- (v2l); \draw (i) -- (v2r); \draw (i) -- (v3l); \draw (i) -- (v3r);
            \draw (i) -- (v4);
            
            \draw (v1) -- (v2l) -- (v3l) -- (v4) -- (v3r) -- (v2r) -- (v1);
            
            \node at (0,-1.5) {$\ast$};
        \end{tikzpicture}
    }
    \subfigure{
        \begin{tikzpicture}[x=1cm,y=.5773cm] 
            \node [circle, fill = black, inner sep = 2pt] (i) at (0,0) {};
            \node [circle, draw = black, inner sep = 2pt] (v1) at (0,-1) {};
            \node [circle, draw = black, inner sep = 2pt] (v2l) at (-.5,-.5) {};
            \node [circle, draw = black, inner sep = 2pt] (v2r) at (.5,-.5) {};
            \node [circle, fill = black, inner sep = .5pt] (v3l) at (-.5,.5) {};
            \node [circle, fill = black, inner sep = .5pt] (v3r) at (.5,.5) {};
            \node [circle, fill = black, inner sep = .5pt] (v4) at (0,1) {};
    
            \draw (i) -- (v1); \draw (i) -- (v2l); \draw (i) -- (v2r); \draw (i) -- (v3l); \draw (i) -- (v3r);
            \draw (i) -- (v4);
            
            \draw (v1) -- (v2l) -- (v3l) -- (v4) -- (v3r) -- (v2r) -- (v1);
            
            \node at (0,-1.5) {$\ast$};
        \end{tikzpicture}
    }\\
    \subfigure{
        \begin{tikzpicture}
            \node[] at (0,0) {case 1};
            \node[] at (1.5,0) {case 2};
            \node[] at (3,0) {case 3};
            \node[] at (4.5,0) {case 4};
        \end{tikzpicture}
    }\\
    \subfigure{
        \begin{tikzpicture}[x=1cm,y=.5773cm] 
            \node [circle, fill = black, inner sep = 2pt] (i) at (0,0) {};
            \node [circle, fill = black, inner sep = .5pt] (v1) at (0,-1) {};
            \node [circle, draw = black, inner sep = 2pt] (v2l) at (-.5,-.5) {};
            \node [circle, fill = black, inner sep = .5pt] (v2r) at (.5,-.5) {};
            \node [circle, draw = black, inner sep = 2pt] (v3l) at (-.5,.5) {};
            \node [circle, fill = black, inner sep = .5pt] (v3r) at (.5,.5) {};
            \node [circle, fill = black, inner sep = .5pt] (v4) at (0,1) {};
    
            \draw (i) -- (v1); \draw (i) -- (v2l); \draw (i) -- (v2r); \draw (i) -- (v3l); \draw (i) -- (v3r);
            \draw (i) -- (v4);
            
            \draw (v1) -- (v2l) -- (v3l) -- (v4) -- (v3r) -- (v2r) -- (v1);
            
            \node at (-.75,-.75) {$\ast$};
        \end{tikzpicture}
    }
    \subfigure{
        \begin{tikzpicture}[x=1cm,y=.5773cm] 
            \node [circle, fill = black, inner sep = 2pt] (i) at (0,0) {};
            \node [circle, draw = black, inner sep = 2pt] (v1) at (0,-1) {};
            \node [circle, draw = black, inner sep = 2pt] (v2l) at (-.5,-.5) {};
            \node [circle, fill = black, inner sep = .5pt] (v2r) at (.5,-.5) {};
            \node [circle, draw = black, inner sep = 2pt] (v3l) at (-.5,.5) {};
            \node [circle, fill = black, inner sep = .5pt] (v3r) at (.5,.5) {};
            \node [circle, fill = black, inner sep = .5pt] (v4) at (0,1) {};
    
            \draw (i) -- (v1); \draw (i) -- (v2l); \draw (i) -- (v2r); \draw (i) -- (v3l); \draw (i) -- (v3r);
            \draw (i) -- (v4);
            
            \draw (v1) -- (v2l) -- (v3l) -- (v4) -- (v3r) -- (v2r) -- (v1);
            
            \node at (-.75,-.75) {$\ast$};
        \end{tikzpicture}
    }
    \subfigure{
        \begin{tikzpicture}[x=1cm,y=.5773cm] 
            \node [circle, fill = black, inner sep = 2pt] (i) at (0,0) {};
            \node [circle, fill = black, inner sep = .5pt] (v1) at (0,-1) {};
            \node [circle, draw = black, inner sep = 2pt] (v2l) at (-.5,-.5) {};
            \node [circle, fill = black, inner sep = .5pt] (v2r) at (.5,-.5) {};
            \node [circle, fill = black, inner sep = .5pt] (v3l) at (-.5,.5) {};
            \node [circle, fill = black, inner sep = .5pt] (v3r) at (.5,.5) {};
            \node [circle, fill = black, inner sep = .5pt] (v4) at (0,1) {};
    
            \draw (i) -- (v1); \draw (i) -- (v2l); \draw (i) -- (v2r); \draw (i) -- (v3l); \draw (i) -- (v3r);
            \draw (i) -- (v4);
            
            \draw (v1) -- (v2l) -- (v3l) -- (v4) -- (v3r) -- (v2r) -- (v1);
            
            \node at (-.75,-.75) {$\ast$};
        \end{tikzpicture}
    }\\
    \subfigure{
        \begin{tikzpicture}
            \node[] at (5,0) {};
            \node[] at (6,0) {case 5};
            \node[] at (7.75,0) {case 6};
            \node[] at (9.5,0) {case 7};
            \node[] at (11,0) {};
        \end{tikzpicture}
    }
    \end{minipage}
    }
    \begin{minipage}{.1\textwidth}
        \subfigure{
        \begin{tikzpicture}
            \draw[->, line width=2pt] (0,0) -- (0,2);
            \node at (.5,1) {$+y$};
        \end{tikzpicture}
    }
    \end{minipage}
    \caption{(a) Naming conventions for neighbourhood of a robot. (b) Cases where a node is \imped. Note that the mirror images of these local states are also \imped.}
    \label{figure:local-states}
\end{figure}
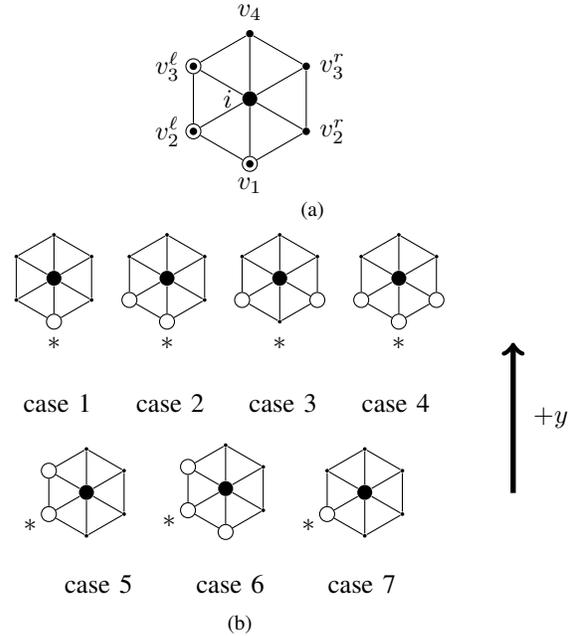

Authors of \cite{Goswami2022} show that the robots do not move out of the bounding polygon. They also show that the visibility graph induced among the robots stay connected, and the dimensions of the bounding polygon reduce with every round. 

\section{GSGS Algorithm \cite{Goswami2022} for \gmrit problem}\label{section:algorithm-gsgs}

In this section, we reword the algorithm in \cite{Goswami2022} to demonstrate its lattice-linearity. 
We define the following macros. 
For a set $L$ of locations around a node $i$, $\textsc{At}(i,L)$ is true iff if there is at least one robot at each location in $L$. $\textsc{Only-At}(i,L)$ is true iff $\textsc{At}(i,L)$, and there is no other robot at locations other than locations in $L$.
A robot $i$ is \textit{extreme} if (1) there is no robot on \textit{top} ($v_4(i)$) of $i$ and (2) if there is a robot on the left ($v_2^\ell$ or $v_3^\ell$) of $i$, then there is no robot on the right ($v_2^r$ or $v_3^r$) of $i$.
If a robot $i$ is extreme, and there is no robot around it, then $i$ is a \textit{terminating} robot.
If $i$ is extreme, and there is a robot only on $v_3(i)$ or there are robots only on both $v_1(i)$ and $v_3(i)$, then $i$ is a \textit{staying} robot.

If $i$ is extreme, and there is a robot on $v_1(i)$ and no robot on $v_3(i)$, then $i$ is a \textit{downward \imped} robot.
If $i$ is not a downward \imped robot, not a staying robot, and not a terminating robot, then it is a \textit{downslant \imped} robot.
If $i$ is not an extreme robot, and there is a robot on both $v_2(i)$ and no robot at its y-coordinate $>$ 0, then $i$ is a \textit{non-extreme \imped} robot.
\begin{center}
    $\begin{array}{|l|}
        \hline 
        \textsc{Extreme}(i)\equiv \lnot\textsc{At}(i,\{v_4\})\land((\textsc{At}(i,\{v_2^r\})\lor\\
        \quad\quad \textsc{At}(i,\{v_3^r\}))\Rightarrow(\lnot \textsc{At}(i,\{v_2^\ell\})\land \lnot\textsc{At}(i,\{v_3^\ell\}))) \\
        \textsc{Terminating}(r)\equiv \textsc{Extreme}(i)\land\\
        \quad\quad (\forall q\in \{v_1,v_2^r,v_3^r,v_4,v_3^\ell,v_2^\ell\}:\lnot\textsc{At}(i,\{q\})).\\
        \textsc{Staying}(i)\equiv \textsc{Extreme}(i)\land (\textsc{Only-At}(i,\{v_3^r\}) \lor\\
        \quad\quad \textsc{Only-At}(i,\{v_3^\ell\}) \lor \textsc{Only-At}(i,\{v_1,v_3^r\}) \lor\\
        \quad\quad \textsc{Only-At}(i,\{v_1,v_3^\ell\})).\\
        \textsc{Downward}(i)\equiv \textsc{Extreme}(i)\land \textsc{At}(i,\{v_1\})\land\\
        \quad\quad \lnot(\textsc{At}(i,\{v_3^r\})\lor \textsc{At}(i,\{v_3^\ell\})).\\
        \textsc{Downslant-Right}(i)\equiv \textsc{Extreme}(i)\land\\
        \quad\quad \lnot \textsc{Downward}(i)\land \lnot\textsc{Staying}(i) \land\\
        \quad\quad \lnot\textsc{Terminating}(r) \land \textsc{At}(i,\{v_2^r\}).\\
        \textsc{Downslant-Left}(i)\equiv \textsc{Extreme}(i)\land\\
        \quad\quad \lnot \textsc{Downward}(i)\land \lnot\textsc{Staying}(i) \land\\
        \quad\quad \lnot\textsc{Terminating}(r) \land \textsc{At}(i,\{v_2^\ell\}).\\
        \textsc{Non-Extreme}(i)\equiv \lnot\textsc{Extreme}(i) \land \textsc{At}(i,\{v_2^r,v_2^\ell\}) \land\\
        \quad\quad \lnot(\textsc{At}(i,\{v_3^r\}) \lor \textsc{At}(i,\{v_3^\ell\}) \lor \textsc{At}(i,\{v_4\})).\\
        \textsc{\Imped-GSGS}(i)\equiv \textsc{Downward}(i)\lor\\
        \quad\quad \textsc{Downslant-Right}(i) \lor\\
        \quad\quad \textsc{Downslant-Left}(i) \lor \textsc{Non-Extreme}(i).\\
        \hline 
    \end{array}$
\end{center}

The algorithm is described as follows.
If a robot $i$ is \textit{downward \imped}, then $i$ moves downwards to $v_1(i)$.
If $i$ is \textit{downslant \imped}, then $i$ moves to $v_2(i)$.
If $i$ is a \textit{non-extreme \imped} robot, then $i$ moves to $v_1(i)$.

\begin{algorithm}\label{algorithm:gsgs-algo}
    Rules for robot $i$.
\end{algorithm}
\begin{center}
    $\begin{array}{|l|}
        \hline 
        \textsc{\Imped-GSGS}(i)\longrightarrow\\
        \begin{cases}
            move(i,v_1(i)) & \text{if $\textsc{Downward}(i)$}\\
            move(i,v_2^r(i)) & \text{if $\textsc{Downslant-Right}(i)$}\\
            move(i,v_2^\ell(i)) & \text{if $\textsc{Downslant-Left}(i)$}\\
            move(i,v_1(i)) & \text{if $\textsc{Non-Extreme}(i)$}
        \end{cases}~\\
        \hline 
    \end{array}$
\end{center}

In \cite{Goswami2022}, authors assume a distributed scheduler. Next, we show that \Cref{algorithm:gsgs-algo} is lattice-linear, Thus, it will be correct even in asynchrony. 

\subsection{Lattice-Linearity}\label{subsection:gsgs-lattice-linearity}

In this subsection, we show lattice-linearity of \Cref{algorithm:gsgs-algo}. Among the lemmas and theorems presented here, \Cref{lemma:lr-layers} and \Cref{lemma:horizontal-layer-moves-down} are adopted from \cite{Goswami2022}. We use them to help prove some properties of \Cref{algorithm:gsgs-algo}. 
All other results show or arise from the lattice-linearity of \Cref{algorithm:gsgs-algo}. 

\begin{lemma}\label{lemma:gmrit-llp}
    The predicate $\forall i:\lnot\textsc{\Imped-GSGS}(i)$ is a lattice-linear predicate on $n$ robots, and the visibility graph does not get disconnected by the actions under \Cref{algorithm:gsgs-algo}.
\end{lemma}

\begin{proof}
    In this proof, we consider the 7 cases as shown in \Cref{figure:local-states} (b) and show that if robot $i$ is \imped, it must execute to reach the goal state.
    We show that if a robot $i$ is \imped, then there exists at least one robot $j$ around $i$ which does not move until $i$ moves. Specifically, \Cref{algorithm:gsgs-algo} imposes that $j$ `waits' for $i$ to move. It means that if $i$ does not move, then the robots cannot find the gathering point.
    
    \textit{Case 1}: This robot $i$ is a downward \imped robot. The other robot that is present below it is not extreme and is also not a non-extreme \imped robot because $i$ is present above it, so it will not move until $i$ changes its location.

    \textit{Case 2}: This robot $i$ is a downward \imped robot. There are two robots, $x_1$ and $x_2$, present at locations $v_1(i)$ and $v_2(i)$ respectively. $x_1$ is not extreme and is also not non-extreme \imped because $i$ is present above it. So $x_1$ will not move until $i$ changes its location. $x_2$ may be \imped.
    $x_2$ can only move to the location of $x_1$ thereby resulting in case 1.
    In this possibility, the robot $i$ remains \imped and its required action does not change.

    \textit{Case 3}: This robot $i$ is a non-extreme \imped robot. There are two robots, $x_1$ and $x_2$, present at locations $v_2(i)$-left and $v_2(i)$-right respectively. For $x_1$ or $x_2$ to be impedensable, there must be some robot at the $v_1(i)$ location, which is not the case, thus, they are not \imped. Hence, $x_1$ and $x_2$ will not move until $i$ changes its location.

    \textit{Case 4}: This robot $i$ is a non-extreme \imped robot. There are three robots, $x_1$, $x_2$ and $x_3$, present at locations $v_2(i)$-left, $v_1(i)$ and $v_2(i)$-right respectively. $x_2$ is not \imped. $x_1$ and $x_3$ may be \imped, based on their local states. If one or both of them move, they will move to the location of $x_2$, resulting in case 2 or case 1.
    In these possibilities, the robot $i$ remains \imped and its required action does not change.

    \textit{Case 5}: This robot $i$ is a downslant \imped robot. There are two robots, $x_1$ and $x_2$, present at locations $v_2(i)$ and $v_3(i)$ respectively. $x_1$ is not extreme and is also not non-extreme \imped. So $x_1$ will not move until $i$ changes its location. $x_2$ may be \imped, based on its local state.
    $x_2$ can move to the location of $x_1$ or $i$ thereby resulting in case 7.
    In this possibility, the robot $i$ remains \imped and its required action does not change.
    
    \textit{Case 6}: This robot $i$ is a downslant \imped robot. There are three robots, $x_1$, $x_2$ and $x_3$, present at locations $v_1(i)$, $v_2(i)$ and $v_3(i)$ respectively. $x_1$ and $x_2$ are not extreme and are also not non-extreme \imped. Initially, $x_1$ and $x_2$ cannot move. $x_3$ may be downward \imped, based on its local state.
    $x_3$ can only move to the location of $x_2$ thereby resulting in case 2. After this, one or both of $x_1$ and $x_2$ can move to the location of $x_1$, resulting in case 1.
    In these possibilities, $i$ remains \imped, its required action may change, but the graph does not get disconnected if it executes under case 6 (using old information, if it assumes, despite the movement of other robots, that it falls in case 6).

    \textit{Case 7}: This robot $i$ is a downslant \imped robot. The other robot $x_1$ that is present at $v_2(i)$ is not extreme and is also not a non-extreme \imped robot, so $x_1$ will not move until $i$ changes its location.

    From these cases, we also have that an \imped robot stays connected to the robots that were its neighbours before it moved. This implies that the visibility graph stays connected after any \imped robots move.
\end{proof}

The robots executing \Cref{algorithm:gsgs-algo}, as shown in \cite{Goswami2022}, stay in the bounding polygon $ABPCD$.
Next, we show, using the above proof, a tighter polygon bounding the robots. To define this polygon, we let $Q$ to be the point such that it is an intersection between the bottom $l2r$ slant of $s$ and the bottom $r2l$ slant of $s$ (cf. \Cref{figure:infinite-triangular-grid}). Let $B'$ be the point of intersection between left layer ($AB$) and the bottom $l2r$ slant of $s$ and let $C'$ be the point of intersection between the right layer ($CD$) and the bottom $r2l$ slant of $s$.
We show that the robots never step out of the polygon $AB'QC'D$, which is tighter than $ABPCD$. 

\begin{observation}\label{observation:imped-5-6}
    If the neighbouring robot, say $j$ of an \imped robot $i$ moves then $i$ or $j$ fall under case 5 or case 6.
\end{observation}

\begin{lemma}\label{lemma:lower-slants}
    Throughout the execution of \Cref{algorithm:gsgs-algo}, the bottom $r2l$ slant and the bottom $l2r$ slant will not change.
\end{lemma}

\begin{proof}
    In a global state $s$, a 
    robot present at the bottom $l2r$ slant of $s$ or a bottom $r2l$ slant of $s$ is represented in cases 5 and 7. From \Cref{algorithm:gsgs-algo} and the proof of \Cref{lemma:gmrit-llp}, if a robot is present at bottom $l2r$ slant (respectively, bottom $r2l$ slant), it will never move below ($v_1(i)$) or left ($v_2^\ell(i)$) of its location (respectively, below ($v_1(i)$) or right ($v_2^r(i)$) of its location).
\end{proof}

\begin{lemma}\label{lemma:lr-layers}\cite{Goswami2022}
    Throughout the execution of \Cref{algorithm:gsgs-algo}, left layer does not move leftwards and right layer does not move rightwards.
\end{lemma}

\begin{lemma}\cite{Goswami2022}\label{lemma:horizontal-layer-moves-down}
    In every round of \Cref{algorithm:gsgs-algo}, the top layer moves at least 1/2 unit in the negative direction of the $y$-axis.
\end{lemma}

\begin{corollary}\label{corollary:never-step-out}(From \Cref{lemma:lower-slants} and \Cref{lemma:lr-layers})
    The robots will never step out of the polygon $AB'QC'D$.
\end{corollary}

\begin{theorem}\label{theorem:gathering-point}
\Cref{algorithm:gsgs-algo} is a lattice-linear self-stabilizing algorithm for the \gmrit problem on $n$ robots executing asynchronously.
\end{theorem}

\begin{proof}
    From \Cref{lemma:lower-slants}, we have that bottom $l2r$ slant and bottom $r2l$ slant do not change. From \Cref{corollary:never-step-out}, we have that the robots will never step out of the polygon $AB'C'DQ$.
    From \Cref{lemma:horizontal-layer-moves-down}, we have that the top layer moves down by at least half a unit in the negative direction of $y$-axis. Thus we have that the robots converge at the point of intersection of the bottom $l2r$ slant and the bottom $r2l$ slant, and the robots will eventually gather at that point.
\end{proof}

\begin{corollary}\label{corollary:gathering-point}(From \Cref{lemma:horizontal-layer-moves-down}, \Cref{corollary:never-step-out} and \Cref{theorem:gathering-point})
    The point, $Q$, where the robots gather, can be uniquely determined from the initial global state.
\end{corollary}  

\subsection{Time Complexity Properties}

In \cite{Goswami2022}, authors showed that \Cref{algorithm:gsgs-algo} converges in $2.5(n+1)$ rounds.
Based on \Cref{corollary:gathering-point} which identifies a predictable gathering point, 
we show that a maximum of $2n$ rounds is sufficient, which is a tighter bound. 

\begin{theorem}\label{theorem:time-gsgs}
    \Cref{algorithm:gsgs-algo} converges in $2n$ rounds.
\end{theorem}

\begin{proof}

We use \Cref{figure:infinite-triangular-grid} to discuss convergence of robots in \Cref{algorithm:gsgs-algo}. As discussed in \Cref{subsection:gsgs-definitions} and \Cref{subsection:gsgs-lattice-linearity}, let $A$, $B'$, $Q$, $C'$ and $D$ be the points obtained by pairwise intersection of the top layer, left layer, bottom $l2r$ slant, bottom $r2l$ slant and right layer. 
Let $h_\ell$ be the depth of the line segment $AB'$, $h_r$ be the depth of the line segment $C'D$, and $w$ be the width of the line segment $AD$. Note that a unit of length of the depth of $AB'$ or $C'D$ is $\sqrt{3}$ times a unit of length of the width of $AD$ due to the geometry of $G$.
Since the robots form a connected graph, $w \leq n$. And, if $w > 0$ then $h_\ell+h_r\leq n$. If $w=0$, we define $h_\ell=0$ and $h_r=n$.

In the case where $w=0$, it can be clearly observed that \Cref{algorithm:gsgs-algo} converges in $n$ rounds. 
Next, we consider if $w>0$. Without the loss of generality, let $h_r\geq h_\ell$. Thus, $h_\ell\leq n/2$. 

Let $E$ be the point of intersection between the bottom $l2r$ slant and the right layer
(cf. \Cref{figure:infinite-triangular-grid}).
We draw a horizontal line (perpendicular to the $y$-axis) through $B'$ and use $J$ to denote its intersection with $DC'$.
Thus, the depth of $AB'JD$ is $h_\ell$. 
Additionally, observe that the length of $B'E$ on the $n$-axis is $w$. Thus the height of $JE$ is $w/2$ units on the $y$-axis.
This means that the depth of $B'EJ$ is $w/2$. By construction of $E$, the depth of $B'QC'J$ is upper bounded by the depth of $B'EJ$. Thus, the depth of $B'QC'J$ cannot exceed $w/2\leq n/2$ units.

Thus, the total depth of $AB'QC'D$ is equal to the sum of the depth of $AB'JD$ and the depth of $B'QC'J$, which is 
upper bounded by $n/2+n/2=n$ units.
From \Cref{lemma:horizontal-layer-moves-down}, the total number of rounds required for the robots to gather is upper bounded by $2n$ moves.
\end{proof} 

\section{Revised Algorithm for \gmrit}\label{section:gsgs-new-algo}

In this section, we present a revised algorithm that simplifies the proof of lattice-linearity. This algorithm is based on the difficulties involved in the proof of \Cref{lemma:gmrit-llp} where we needed to consider the possible actions taken by the neighbours of an \imped robot. Our proof would have been simpler
if all the neighbours of an \imped robot $i$ would not be allowed to move until $i$ moves.
Additionally, from \Cref{observation:imped-5-6}, we have that if a robot $j$, neighbouring to an \imped node $i$, is \imped, then $i$ or $j$ fall in case $5$ or case $6$.

These issues can be alleviated by removing cases 5 and 6 from the algorithm.
The macros that we utilize are as follows. A robot is \textit{downward \imped} if its local state is one of those represented in cases 1, 2, 3 or 4 (and their mirror images; cf \Cref{figure:local-states} (b)). A robot is is \textit{downslant \imped} if its local state is that represented in case 7.
$$
\begin{array}{|l|}
    \hline 
    \textsc{Downward-II}(i)\equiv (\textsc{At}(i,\{v_1\})\land \lnot\textsc{At}(i,\{v_3^\ell, v_4, v_3^r\}))\lor\\
    \quad\quad \textsc{Only-At}(i,\{v_2^\ell, v_2^r\}).\\
    \textsc{Downslant-Left-II}(i)\equiv \textsc{Only-At}(i,\{v_2^\ell\}).\\
    \textsc{Downslant-Right-II}(i)\equiv \textsc{Only-At}(i,\{v_2^r\}).\\
    \hline 
\end{array}
$$

The revised algorithm is as follows. A downward impedensable robot moves to $v_1(i)$ location, and a downslant \imped robot moves to $v_2(i)$ location.

\begin{algorithm}\label{algorithm:gsgs-new-algo}
    Rules for robot $i$.
\end{algorithm}
\begin{center}
    $\begin{array}{|l|}
        \hline 
        \textsc{Downward-II}(i)\longrightarrow move(i,v_1(i)).\\
        \textsc{Downslant-Right-II}(i)\longrightarrow move(i,v_2^r(i)).\\
        \textsc{Downslant-Left-II}(i)\longrightarrow move(i,v_2^\ell(i)).\\
        \hline 
    \end{array}
    $
\end{center}

In \Cref{algorithm:gsgs-new-algo}, because of the removal of cases 5 and 6, any robot around an \imped robot does not move. Thus lattice-linearity of this algorithm can be visualized more intuitively. Consequently, we have the following lemma.

\begin{lemma}
    The predicate $\forall i: \lnot(\textsc{Downward-II}(i)$ $\lor$
        $\textsc{Downslant-Right-II}(i)$ $\lor$
        $\textsc{Downslant-Left-II}(i))$,
    is a lattice-linear predicate on $n$ robots, and the visibility graph does not get disconnected by the actions under \Cref{algorithm:gsgs-new-algo}.
\end{lemma}

\Cref{lemma:horizontal-layer-moves-down} shows the top-layer moves down in each round. This proof is not affected by the removal of cases 5 and 6, as the robot executing in cases 5 or 6 is \textit{not} a top-layer robot. 
Consequently, \Cref{algorithm:gsgs-new-algo} follows the properties as described in \Cref{lemma:horizontal-layer-moves-down}, \Cref{theorem:gathering-point} and hence \Cref{theorem:time-gsgs}. 
Thus, we obtain the following theorem.

\begin{theorem}
    \Cref{algorithm:gsgs-new-algo} is a lattice-linear self-stabilizing algorithm for \gmrit problem on $n$ robots executing asynchronously. It converges in $2n$ rounds, and the robots gather at $Q$.
\end{theorem}

\section{Related Work}\label{section:literature}

\noindent\textit{Lattice-Linearity}: 
Lattice-linearity was introduced in \cite{Garg2020}. Additionally, in \cite{Garg2021,Garg2022}, the authors studied lattice-linearity in, respectively, housing market problem and several dynamic programming problems. We call the problems studied in these papers \textit{lattice-linear problems}.
We introduced eventually lattice-linear algorithms \cite{Gupta2021} and fully lattice-linear algorithms \cite{Gupta2022,Gupta2023} for non-lattice-linear problems.

\noindent\textit{Robot gathering problem on discrete grids}:
In a general case, it is impossible to gather a system of robots if their visibility graph is not a connected graph. One-axis agreement and distance-1 myopia are the minimal capabilities that robots need to converge on a triangular grid \cite{Goswami2022}.

A system of robots with minimal capabilities has been studied with several output requirements, including gathering \cite{Goswami2022,DAngelo2016,Flocchini2005}, dispersion \cite{Augustine2018}, arbitrary pattern formation \cite{Bose2020}. Gathering of robots has been studied more recently in \cite{Bhagat2015,Klasing2010}. We focus on systems of robots on grids, mainly the papers that study gathering.

Robots placed on an infinite rectangular grid were studied in \cite{Poudel2021}, where the authors presented two algorithms for gathering. A synchronous scheduler is assumed and the robots require, respectively, distance-2 and distance-3 visibility.
Moreover, under the latter algorithm, robots may not gather at one point but will gather a horizontal line segment of unit length.

Robots placed on an infinite triangular grid were studied in \cite{Cicerone2021}, where the authors provided an algorithm to form any arbitrary pattern.
They require full visibility.
Their algorithm works only when the the initial global state is asymmetric.
Authors of \cite{Shibata2021} have studied gathering problem of 7 robots -- initially, 6 of them form a hexagon and one robot is present at the centre of that hexagon. They require the initial state to form a connected visibility graph; the system finally reaches a global state where the maximum distance between two robots is minimized. A synchronous scheduler is assumed.
In \cite{DAngelo2016}, authors  characterized the problem of gathering on a tree and finite grid.

This paper studies an algorithm \cite{Goswami2022} that converges assuming that the robots are myopic, and can use a unidirectional camera, that sees one neighbour at a time. The robots form an arbitrary connected graph initially.

\noindent\textit{Abstractions in Concurrent Computing}: 
Since this paper focuses on asynchronous computations, we also study other abstractions in the context of concurrent systems: non-blocking (lock-free/wait-free), starvation-free and serializability.

An algorithm is \textit{non-blocking} if in a system running such algorithm, if a node fails or is suspended, then it does not result in failure or suspension of another node. 
A non-blocking algorithm is \textit{lock-free} if system-wide progress can be guaranteed. For example, algorithms for implementing lock-free singly-linked lists and binary search tree are, respectively, presented in \cite{Valois1995} and \cite{Natarajan2014}. In such systems, if a read/write request is blocked then other nodes continue their actions normally.

A non-blocking algorithm is \textit{wait-free} if progress can be guaranteed per node. A wait-free sorting algorithm is studied in \cite{Shavit1997}, which sorts an array of size $N$ using $n\leq N$ computing nodes, and an $O(n)$ time wait-free approximate agreement algorithm is presented in \cite{Attiya1994}. In such systems, in contrast to lock-free systems, it must be guaranteed that all nodes make progress individually.

A key difference between non-blocking algorithms and asynchronous algorithms is the \textit{system-perspective} for which they are designed. To understand this, observe that from a perspective, the lattice-linear, asynchronous, algorithm considered in this paper is wait-free. Each robot checks the location of its neighbouring robots. Then, it executes an action, if it is enabled, without synchronization. More generally, in an asynchronous algorithm, each node reads the state of its relevant neighbours to check if the guard evaluates to true. It can, then, update its state without coordination with other nodes. 

That said, the goal of asynchronous algorithms is not the progress/blocking of individual nodes 
(e.g., success of insert request in a linked list and a binary search tree, respectively, in \cite{Valois1995} and \cite{Natarajan2014}).
Rather it focuses on the progress from the perspective of the system, i.e., the goal is not about the progress of an action by a robot/node but rather that of the entire system.  
For example, under the GSGS algorithm, if one of the robots is slow or does not move, the robots will not converge. 
However, the robots can run without any coordination and they can use a unidirectional camera, in contrast to using an omnidirectional camera or taking a snapshot of the entire arena.
In fact, the notion of \imped (recall that we rewrote the GSGS algorithm such that, in any global state, all enabled robots are \imped) captures this. An \imped robot has to make progress in order for the system to make progress.

Wait-free computations are also studied from a system perspective as well, e.g., in \cite{Castaneda2018}.  
They assume a finite grid, ability of a robot to observe the entire arena, and ability to communicate with other robots (via messages or shared registers). 
Even under these assumptions, they show that the problem is unsolvable if all we know is that the graph is connected and a robot can crash. 
Our paper demonstrates that asynchronous execution is possible without the ability to communicate or observe the entire arena
subject to the constraint that every enabled robot eventually makes a move.

\textit{Starvation} happens when requests of a higher priority prevent a request of lower priority from entering the critical section indefinitely. To prevent starvation, algorithms are designed such that the priority of pending requests are increased dynamically. Consequently, a low-priority request eventually obtains the highest priority. Such algorithms are called \textit{starvation-free} algorithms. For example, authors of \cite{Kim2005} and \cite{Attiya2010}, respectively, present a starvation-free algorithm to schedule queued traffic in a network switch and a starvation-free distributed directory algorithm for shared objects.
Asynchronous algorithms are starvation-free, as long as all enabled processes can execute. If all enabled processes can execute, convergence is guaranteed.

\textit{Serializability} allows only those executions to be executed concurrently which can be modelled as some permutation of a sequence of those executions. In other words, serializability does not allow nodes to read and execute on old information of each other: only those executions are allowed in concurrency such that reading fresh information, as if the nodes were executing in an interleaving fashion, would give the same result. Serializability is heavily utilized in database systems, and thus, the executions performed in such systems are called \textit{transactions}. Authors of \cite{Papadimitriou1979} show that corresponding to several transactions, determining whether a sequence of read and write operations is serializable is an NP-Complete problem. They also present some polynomial time algorithms that approximate such serializability. Authors of \cite{Fle1982} consider the problem in which the sequence of operations performed by a transaction may be repeated infinitely often. They describe a synchronization algorithm allowing only those schedules that are serializable in the order of commitment.

The asynchronous execution considered in this paper is not \textit{serializable}, especially, since the reads can be from an old global state. Even so, the algorithm converges, and does not suffer from the overhead of synchronization required for serializability.

\section{Conclusion}\label{section:conclusion}

In this paper, we show lattice-linearity of the algorithm developed by Goswami et al \cite{Goswami2022} for gathering robots on an infinite triangular grid. This removes the assumptions of synchronization from the algorithm and thus makes a system running this algorithm fully tolerant to asynchrony. We also present a revised algorithm that simplifies the proof of lattice-linearity without losing any of the desired properties (e.g., convergence time, stabilization).

Lattice-linearity implies that the locations, possibly visited by a robot, form a total order. The total order is a result of the fact that we are able to determine all and the only robots in any global state that are \imped, and an \imped robot has only one choice of action. By making this observation, it can also be noticed that we can closely predict the executions that the robots would perform. As a result, we are able to (1) compute the exact arena traversed by the robots throughout the execution of the algorithm (\Cref{lemma:lower-slants} and \Cref{lemma:lr-layers}), and (2) deterministically predict the point of gathering of the robots (\Cref{corollary:gathering-point}).

We also provided a better upper bound on the time complexity of this algorithm. Specifically, we show that it converges in $2n$ rounds, whereas \cite{Goswami2022} showed that a maximum of $2.5(n+1)$ rounds are required.
This was possible due to the observations that followed from the proof of lattice-linearity of this algorithm. 

Being able to design systems that fully tolerate asynchrony has been a desired, albeit unattainable, objective. Lattice-linear systems provide with a deterministic, discrete, guarantee to attain such fault tolerance. Studying whether existing algorithms exploit lattice-linearity is extremely interesting and beneficial. Specifically, it allows us to eliminate costly, sophisticated, assumptions of synchronization from existing systems, instead of redesigning them from scratch.

\bibliography{gsgs.bib}
\bibliographystyle{acm}

\end{document}